\documentclass[conference]{IEEEtran}
\ifCLASSINFOpdf
\else
\fi

\usepackage[utf8]{inputenc} 
\usepackage[T1]{fontenc}
\usepackage{url}
\usepackage{ifthen}
\usepackage{cite}
\usepackage[cmex10]{amsmath} % Use the [cmex10] option to ensure complicance
\usepackage{amssymb}
\usepackage{mathtools}
\usepackage{amsfonts}
\usepackage{graphicx}
\usepackage{MnSymbol,wasysym}
\usepackage{braket}
\usepackage{algorithm}
\usepackage{mathrsfs}
\usepackage{algorithmic}

\usepackage{array}
\newcolumntype{P}[1]{>{\centering\arraybackslash}p{#1}}
\newcolumntype{M}[1]{>{\centering\arraybackslash}m{#1}}

\newcommand{\logdet}{\log\det}
\usepackage{amsthm}
\usepackage{optidef}

\newtheorem{theorem}{Theorem}[section]

\newtheorem{lemma}[theorem]{Lemma}

\theoremstyle{remark}

\theoremstyle{definition}

\begin{document}
\title{Rate of Prefix-free Codes in LQG Control Systems with Side Information} 

% %%% Single author, or several authors with same affiliation:
 %\author{%
  % \IEEEauthorblockN{Travis C.~Cuvelier & Takashi Tanaka}
   %\IEEEauthorblockA{ETH Zürich\\
    %                 ISI (D-ITET)\\
     %%               Email: moser@isi.ee.ethz.ch}
% }

%%% Several authors with up to three affiliations:
\author{
  \IEEEauthorblockN{Travis C. Cuvelier}
  \IEEEauthorblockA{Department of Electrical and Computer Engineering\\
                    The University of Texas at Austin, 
                    Austin, TX USA\\
                   Email: tcuvelier@utexas.edu}
  \and
  \IEEEauthorblockN{Takashi Tanaka}
  \IEEEauthorblockA{Department of Aerospace Engineering and Engineering Mechanics\\
                    The University of Texas at Austin, 
                    Austin, TX USA\\
                    Email: ttanaka@utexas.edu }
}
\maketitle

\begin{abstract}
     In this work, we study an LQG control system where one of two feedback channels is discrete and incurs a communication cost. We assume that a decoder (co-located with the controller) can make noiseless measurements of a subset of the state vector (referred to as \textit{side information}) meanwhile a remote encoder (co-located with a sensor) can make arbitrary measurements of the entire state vector, but must convey its measurements to the decoder over a noiseless binary channel. Use of the channel incurs a communication cost,  
     quantified as the time-averaged expected length of prefix-free binary codeword. We study the tradeoff between the communication cost and control performance. The formulation motivates a constrained directed information minimization problem, which can be solved via convex optimization. Using the optimization, we propose a quantizer design and a subsequent achievability result. 
\end{abstract}
\section{Introduction}
In this work we consider discrete-time MIMO LQG control in a system where some measurements incur a communication cost, but others do not. As in \cite{SDP_DI} and \cite{tanakaISIT}, we study the tradeoff between control performance and communication cost, where the latter is measured in terms of the average length of prefix-free codewords. Our principal motivation is a sensing scenario where an energy constrained remote platform (the encoder) must encode, and then wirelessly transmit, its measurements to a joint fusion center/controller (decoder) which contains some sensors of its own. We model the decoder measurements as noiseless observations of a subset of the state vector indices, which we refer to as \textit{side information} (SI). We consider a setup where both the encoder and decoder have access to the decoder's measurements. In the remote sensing scenario, it may be reasonable to assume that the decoder has sufficient energy to feed its measurements back to the encoder while the sensor platform could be constrained-- under some additional assumptions, minimizing the time-averaged bitrate from the encoder to decoder is a surrogate for minimizing the energy the sensor platform ``spends" on communication.  We
establish a converse bound on the minimum prefix-free codeword length in terms of Massey's directed information (DI) \cite{masseyDI}. The bound applies to the case when the SI is known at both the encoder and decoder, and thus applies when the SI is known at the decoder only. The converse motivates a rate distortion problem where a DI term is minimized subject to a constraint on control performance. The problem is solved optimally via a tractable mathematical program (namely a log-determinant optimization)\cite{oronNew}. We use the optimization to derive an achievability result based on the construction in \cite{tanakaISIT}. 

Massey's DI quantifies the flow of information from one stochastic process to another\cite{masseyDI}. In \cite{silvaFirst}, the time-averaged bitrate of a prefix-free codec inserted into the feedback loop of a SISO control system was shown to be lower bounded by the DI from the plant output to the control input. Also, \cite{silvaFirst} motivated the use of entropy dithered quantization (EDQ) in control systems subject to data rate constraints. Extending \cite{silvaFirst} to the MIMO setting, \cite{SDP_DI} motivated a rate distortion problem that minimized DI in an LQG control system subject to a constraint on performance. Under standard linear/Gaussian plant dynamics, \cite{SDP_DI} showed that any optimal measurement and control policy could be implemented via a three-stage separation architecture; namely a linear/Gaussian sensor, a Kalman filter, a certainty equivalence linear feedback controller. The optimization to find the minimum DI (and the minimizing policy) was formulated as a semidefinite (log-determinant) program \cite{SDP_DI}. \cite{tanakaISIT} gave operational significance to the minimal DI (and minimizing policy) in \cite{SDP_DI}. In \cite{tanakaISIT}, it was shown that a zero-delay source coding scheme, based on quantizing Kalman filter innovations via EDQ, followed by prefix-free coding achieves a DI cost within $\frac{n}{2}\log(\frac{\pi e}{3})+2$ bits of the minimal DI cost in \cite{SDP_DI}. Likewise, \cite{kostinaTradeoff} studied the tradeoff between DI and LQG performance, proved converse bounds applying to plants with non-Gaussian disturbances, and demonstrated achievability without dithering. 

The impact of SI (modeled as a decoder-side linear observation of the state vector in additive Gaussian noise) on the tradeoff between DI and LQG performance in LTI SISO systems was investigated in \cite{kostinaSI}. It was argued that it suffices to consider a rate distortion problem in a related tracking problem and that linear/Gaussian policies were optimal. In \cite{photisSI}, an optimization problem was formulated to analyze the minimum attainable DI in a MIMO LQG control system with SI assuming linear feedback policies. Very recently, \cite{oronNew} proved that linear/Gaussian policies conforming to the ``three-stage separation" architecture of \cite{SDP_DI} achieve optimal performance in a MIMO time-varying generalization of the original control problem posed in \cite{kostinaSI}. It is also argued that it suffices to consider time-invariant policies in the time-invariant infinite horizon setting\cite{oronNew}. In \cite{kostinaSI}, \cite{photisSI}, and \cite{oronNew}, SI at the encoder does not impact the rate-distortion tradeoff.

In this work, our system model differs slightly from that in \cite{photisSI} and is slightly less general than that in \cite{oronNew}. Our perspective is quite different. We motivate our rate distortion problem, and demonstrate an achievability result, in terms of \textit{digital} communications. The achievability approaches in \cite{photisSI} and \cite{oronNew} are \textit{analog} in the sense that the feedback channel is continuous. Our contributions are summarized as follows:
\begin{enumerate}
    \item Assuming that the feedback channel from the encoder to decoder is binary and noiseless, we derive a lower bound on the minimum expected prefix-free codeword length under a constraint on control performance. The converse result motivates a rate distortion formulation. 
    \item Via the three-stage separation principle (cf. \cite[(19)]{oronNew}),  we derive a semidefinite program equivalent to the  rate distortion problem. \footnote{When we originally submitted this manuscript, we proposed three-stage separation as a conjecture. After submitting, we became aware that it was shown to optimal in the commensurately published \cite{oronNew}. We derived our SDP formulation independently, and provide additional system theoretic commentary with respect to \cite{oronNew}.} 
    \item We provide a recipe to design both a sensor and a quantizer that nearly achieves the performance of the rate distortion formulation. Namely, we specify both a zero-delay quantizer design and source coding protocol.
\end{enumerate}A version of this paper with appendices is provided in \cite{our_appendix}. 

\textbf{Notation}: 
We denote scalars by lower case letters $s$, vectors by boldface lower-case letters $\mathbf{v}$, and matrices by boldface capitols $\mathbf{M}$. $\mathbf{M}^\mathrm{T}$ denotes transpose. We use $x_{1:t}$  to denote the sequence $(x_1,x_2,\dots, x_t)$, and $\{x_{t}\}$ for $x_{1:\infty}$. We define the ``time shifted" sequence $\mathbf{x}^+_{1:t} = (0, \mathbf{x}_{1},\dots,\mathbf{x}_{t-1})$. If ${t}<1$, $\mathbf{x}_{1:t} = \emptyset$. Denote the set of finite length binary strings $\{0,1\}^*$. Denote the entropy of a discrete random variable (RV) $H$, differential entropy by $h$, and mutual information (MI) by $I$. Denote causally conditioned DI 
\begin{align}\label{eq:DIDEF}
    I(\mathbf{p}_{1:T}\rightarrow\mathbf{q}_{1:T}|| \mathbf{r}_{1:T}) = \sum\nolimits_{t=1}^{T}  I(\mathbf{p}_{1:t};\mathbf{q}_{t} | \mathbf{q}_{1:t-1}, \mathbf{r}_{1:t}).
\end{align} If $A,B,C$ are RVs and $A$ is independent of $C$ given $B$ we say that $A$, $B$, $C$ form a Markov chain and write $A\leftrightarrow B \leftrightarrow C$. 
\section{System model and problem formulation}\label{sec:formulation}
 Fig. \ref{fig:flowchart} illustrates our assumed system model. We assume a MIMO plant, a generally randomized sensor/encoder, and two feedback channels (one for SI and one for prefix-free codewords) from the encoder to a possibly randomized decoder/controller. Let $\mathbf{x}^1_t\in \mathbb{R}^{n}$ and $\mathbf{x}^2_t\in \mathbb{R}^{m}$. The state vector is defined as $\mathbf{x}_t = [(\mathbf{x}^1_t)^{\mathrm{T}},(\mathbf{x}^2_t)^{\mathrm{T}}]^{\mathrm{T}}$. Let $\mathbf{A}_{11}\in \mathbb{R}^{n\times n}$, $\mathbf{A}_{12}\in \mathbb{R}^{n\times m}$, $\mathbf{A}_{21}\in \mathbb{R}^{m\times n}$, and $\mathbf{A}_{22}\in \mathbb{R}^{m\times m}$ be block partitions of the system matrix $\mathbf{A}$, and define $\mathbf{W}_{11}\in \mathbb{R}^{n\times n}$, and  $\mathbf{W}_{22}\in \mathbb{R}^{m\times m}$. The plant dynamics are given by
\begin{subequations}
\begin{align}\label{eq:dynamics}
   \begin{bmatrix}
    \mathbf{x}^1_{t+1}\\\mathbf{x}^2_{t+1}
    \end{bmatrix} = \begin{bmatrix} \mathbf{A}_{11} &  \mathbf{A}_{12} \\ \mathbf{A}_{21} & \mathbf{A}_{22}\end{bmatrix}\begin{bmatrix}
    \mathbf{x}^1_{t}\\\mathbf{x}^2_{t}
    \end{bmatrix}+ \mathbf{B}\mathbf{u}_t+\mathbf{w}_t\text{, where }
\end{align} 
\begin{align}
\mathbf{w}_t\sim\mathcal{N}(\mathbf{0},\mathbf{W})\text{ and } \mathbf{W}=\begin{bmatrix} \mathbf{W}_{11} &  \mathbf{0} \\ \mathbf{0} & \mathbf{W}_{22}\end{bmatrix}.
\end{align}
\end{subequations}
We assume $\mathbf{W}_{11},\mathbf{W}_{22}\succ 0$ and the $\mathbf{w}_t$ are IID. We assume $\mathbf{B}\in\mathbb{R}^{n+m\times u}$ and that $(\mathbf{A},\mathbf{B})$ is stabilizable. \begin{figure}[]
	\centering
	\includegraphics[scale = .22]{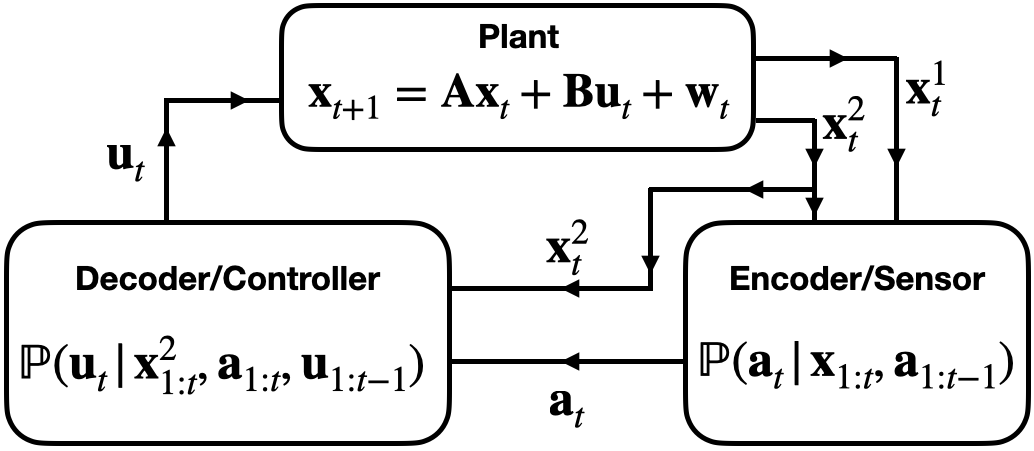}
    \vspace{-.4cm}
	\caption{The encoder has access to $\mathbf{x}^1_{t}$ and $\mathbf{x}^2_{t}$, while the decoder can access $\mathbf{x}^2_{t}$, only.  At every time $t$, the encoder transmits a prefix-free codeword $\mathbf{a}_{t}\in \{0,1\}^*$ to the controller. As in \cite{SDP_DI}\cite{tanakaISIT}, the length of the codeword provides a notion of communication cost. Intuitively, the decoder relies on a discrete channel to convey any knowledge of $\{\mathbf{x}^1_{t}\}$ not contained in  $\{\mathbf{x}^2_{t}\}$ to the decoder. The decoder generates the control input $\mathbf{u}_{t}$}
	\label{fig:flowchart}
\end{figure} The sensor/encoder policy is a sequence of causally conditioned stochastic kernels denoted $
    \mathbb{P}(\mathbf{a}_{1:\infty}|| \mathbf{x}_{1:\infty})=\{\mathbb{P}(\mathbf{a}_t| \mathbf{x}_{1:t},\mathbf{a}_{1:t-1})\}_{t=1,\dots}$,  the decoder/controller policy is defined analogously and denoted $
    \mathbb{P}(\mathbf{u}_{1:\infty}|| \mathbf{a}_{1:\infty},\mathbf{x}^2_{1:\infty})=  \{\mathbb{P}(\mathbf{u}_t|\mathbf{a}_{1:t},\mathbf{x}^2_{1:t},\mathbf{u}_{1:t-1})\}_{t=1,2,\dots}$.

Let $\ell(\mathbf{a}_{t})$ be the length of the codeword $\mathbf{a}_{t}\in \{0,1\}^*$ (in bits). We seek policies  that minimize the time averaged expected codeword length subject to a constraint on control performance. Following  from \cite{tanakaISIT}, we pursue the optimization:
\begin{equation}
\begin{aligned}
& \underset{\substack{\mathbb{P}(\mathbf{a}_{1:{\infty}}||\mathbf{x}_{1:\infty}) \\ \mathbb{P}(\mathbf{u}_{1:\infty}|| \mathbf{a}_{1:\infty},\mathbf{x}^2_{1:\infty})}}{\inf}  \underset{T\rightarrow \infty}{\lim\sup}\text{ }\frac{1}{T}\sum\nolimits_{t=1}^{T}\mathbb{E}[\ell(\mathbf{a}_{t})] \\ &\text{s.t. }   \underset{T\rightarrow \infty}{\lim\sup}\frac{1}{T}\sum\nolimits_{t=1}^{T}\mathbb{E}[\lVert \mathbf{x}_{t+1} \rVert_{\mathbf{Q}}^{2} +\lVert \mathbf{u}_{t} \rVert_{\mathbf{R}}^{2}] \le \gamma
\end{aligned} 
\end{equation} where $\mathbf{Q}\succ 0$, $\mathbf{R}\succ 0$. The expectations are taken with respect to the joint measure induced by the policies and plant dynamics. 

\section{Converse}\label{sec:converse}
The converse follows from \cite{tanakaISIT} and \cite{SDP_DI}. 
\begin{theorem}[A converse proof]\label{theorem:converse}
Consider the model in Fig. \ref{fig:flowchart}. Let $\ell(\mathbf{a}_{t})$ be the length of the codeword $\mathbf{a}_t$ in bits. For any (possibly randomized) control and encoding/decoding policies, we have
\begin{align}\label{eq:converse_result}
    \sum\nolimits_{t=1}^{T}\mathbb{E}[\ell(\mathbf{a}_{t})]\ge I(\mathbf{x}_{1:T} \rightarrow \mathbf{a}_{1:T}||\mathbf{x}^2_{1:T}).
\end{align}
\end{theorem} 
\begin{proof}
The model assumes that $\mathbf{a}_{t}$ is a codeword from a prefix-free code. Let $\mathcal{A}_{t} = \{\mathsf{a}\in \{0,1\}^* : \mathbb{P}(\mathbf{a}_{t}=\mathsf{a})>0\}$. At every time $t$,   if $\mathsf{a}_{1},\mathsf{a}_{2}\in \mathcal{A}_{t}$ the prefix-free assumption guarantees that  $\mathsf{a}_{1}$ is not a prefix of $\mathsf{a}_{2}$ and vice-versa. We claim that 
\begin{align}\label{eq:trivialNontrivial}
 \mathbb{E}[\ell(\mathbf{a}_{t})] \ge H(\mathbf{a}_{t}),
\end{align} this follows from a claim that for every $t$, any function $C_{t}: \{0,1\}^*\rightarrow \{0,1\}^*$  satisfying
$\mathsf{a}=C_{t}(\mathsf{a})$ for all $\mathsf{a}\in\mathcal{A}_{t}$ is a prefix-free code (in the terminology of \cite[Ch. 5]{elemIT} ) from $\mathcal{A}_{t}$ to $\{0,1\}^*$. For any prefix-free code $C^*_{t}$ (cf. \cite[Theorem 5.3.1]{elemIT}) 
\begin{align}
    \mathbb{E}[\ell(C^*_{t}(\mathbf{a}_{t}))] \ge H(\mathbf{a}_{t}). 
\end{align} Since $C_{t}$ is identity on $\mathcal{A}_{t}$, we have $\mathbb{E}[\ell(C_{t}(\mathbf{a}_{t}))] = \mathbb{E}[\ell(\mathbf{a}_{t})]$, and (\ref{eq:trivialNontrivial}) follows. We discuss (\ref{eq:trivialNontrivial}) in \cite[Appendix \ref{app:SIEO}]{our_appendix}. 

At every time $t$ we have the following chain of inequalities
\begin{IEEEeqnarray}{rCl} \mathbb{E}[\ell(\mathbf{a}_{t})] &\ge& H(\mathbf{a}_{t})\label{eq:trivialNontrivial2} \\ &\ge&H(\mathbf{a}_t | \mathbf{a}_{1:t-1}, \mathbf{x}^2_{1:t}) \label{eq:conditioningReducesEntropy} \\ &\ge& H(\mathbf{a}_t | \mathbf{a}_{1:t-1}, \mathbf{x}^2_{1:t})-  H(\mathbf{a}_t | \mathbf{a}_{1:t-1},  \mathbf{x}_{1:t})\label{eq:vanillaEntropyPositive},
\end{IEEEeqnarray} Note that (\ref{eq:trivialNontrivial2}) is precisely (\ref{eq:trivialNontrivial}), (\ref{eq:conditioningReducesEntropy}) follows since conditioning reduces entropy, and (\ref{eq:vanillaEntropyPositive}) follows since discrete entropy is positive. The right hand side of (\ref{eq:vanillaEntropyPositive}) is  $I(\mathbf{a}_t;\mathbf{x}_{1:t}|\mathbf{a}_{1:t-1},\mathbf{x}^2_{1:t})$. 
Summing over $t$, and applying (\ref{eq:DIDEF}) gives (\ref{eq:converse_result}). 
\end{proof}
\section{Rate Distortion Formulation}
Given the converse in Sec. \ref{sec:converse}, the arguments in  \cite{SDP_DI} and \cite{tanakaISIT} suggest attempting  the following optimization 
\begin{equation}\label{eq:desired}
\begin{aligned}
& \underset{\substack{\mathbb{P}(\mathbf{a}_{1:\infty}|| \mathbf{x}_{1:\infty}) \\ \mathbb{P}(\mathbf{u}_{1:\infty}|| \mathbf{a}_{1:\infty},\mathbf{x}^2_{1:\infty})}}{\inf}  \underset{T\rightarrow \infty}{\lim\sup}\text{ }\frac{1}{T} I(\mathbf{x}_{1:T} \rightarrow \mathbf{a}_{1:T}||\mathbf{x}^2_{1:T})  \\ &\text{s.t. }   \underset{T\rightarrow \infty}{\lim\sup}\frac{1}{T}\sum\nolimits_{t=1}^{T}\mathbb{E}[\lVert \mathbf{x}_{t+1} \rVert_{\mathbf{Q}}^{2} +\lVert \mathbf{u}_{t} \rVert_{\mathbf{R}}^{2}] \le \gamma,
\end{aligned} 
\end{equation} where the infimum is over all possible encoder and decoder policies and all expectations are computed under the measure induced by the policies and the plant dynamics. Let $\{\mathbf{y}_{t}\}$ denote a sequence of (not necessarily discrete) random variables. Define the set of causally conditioned kernels $\mathbb{P}(\mathbf{y}_{1:\infty}|| \mathbf{x}_{1:\infty})=\{\mathbb{P}(\mathbf{y}_t| \mathbf{x}_{1:t},\mathbf{y}_{1:t-1})\}_{t=1,\dots}$ and  $
    \mathbb{P}(\mathbf{u}_{1:\infty}|| \mathbf{y}_{1:\infty},\mathbf{x}^2_{1:\infty})=  \{\mathbb{P}(\mathbf{u}_t|\mathbf{y}_{1:t},\mathbf{x}^2_{1:t},\mathbf{u}_{1:t-1})\}_{t=1,2,\dots}$. The infimum in (\ref{eq:desired}) is lower bounded by 
\begin{equation}\label{eq:desired_continuous}
\begin{aligned}
& \underset{\substack{\mathbb{P}(\mathbf{y}_{1:\infty}|| \mathbf{x}_{1:\infty}) \\ \mathbb{P}(\mathbf{u}_{1:\infty}|| \mathbf{y}_{1:\infty},\mathbf{x}^2_{1:\infty})}}{\inf}  \underset{T\rightarrow \infty}{\lim\sup}\text{ }\frac{1}{T} I(\mathbf{x}_{1:T} \rightarrow \mathbf{y}_{1:T}||\mathbf{x}^2_{1:T})  \\ &\text{s.t. }   \underset{T\rightarrow \infty}{\lim\sup}\frac{1}{T}\sum\nolimits_{t=1}^{T}\mathbb{E}[\lVert \mathbf{x}_{t+1} \rVert_{\mathbf{Q}}^{2} +\lVert \mathbf{u}_{t} \rVert_{\mathbf{R}}^{2}] \le \gamma.
\end{aligned} 
\end{equation} That (\ref{eq:desired_continuous}) lower bounds (\ref{eq:desired}) follows from expanding the domain of minimization. In (\ref{eq:desired}) the optimization is restricted to kernels where $\mathbf{a}_{t}$ is a discrete codeword, whereas in (\ref{eq:desired_continuous}) we make no such assumption. 

Note that (\ref{eq:desired_continuous}) is an optimization over an infinite dimensional policy space and is not computationally amenable. Recently, \cite{oronNew} demonstrated that the minimum in (\ref{eq:desired_continuous}) is achievable by a time invariant linear/Gaussian policy conforming to the three-stage separation architecture depicted in Fig. \ref{fig:conjecture}; namely, the optimal policy consists of a time-invariant linear-Gaussian sensor, a Kalman filter, and a certainty equivalence controller. Such a structural result allows us to convert (\ref{eq:desired_continuous}) into an equivalent finite dimensional optimization. We discuss the optimal architecture in the following subsection. 
\subsection{Three stage test channel (cf. \cite{SDP_DI}, \cite{oronNew})}
The feedback loop contains three components: 
\begin{figure}[]
	\centering
	\includegraphics[scale = .23]{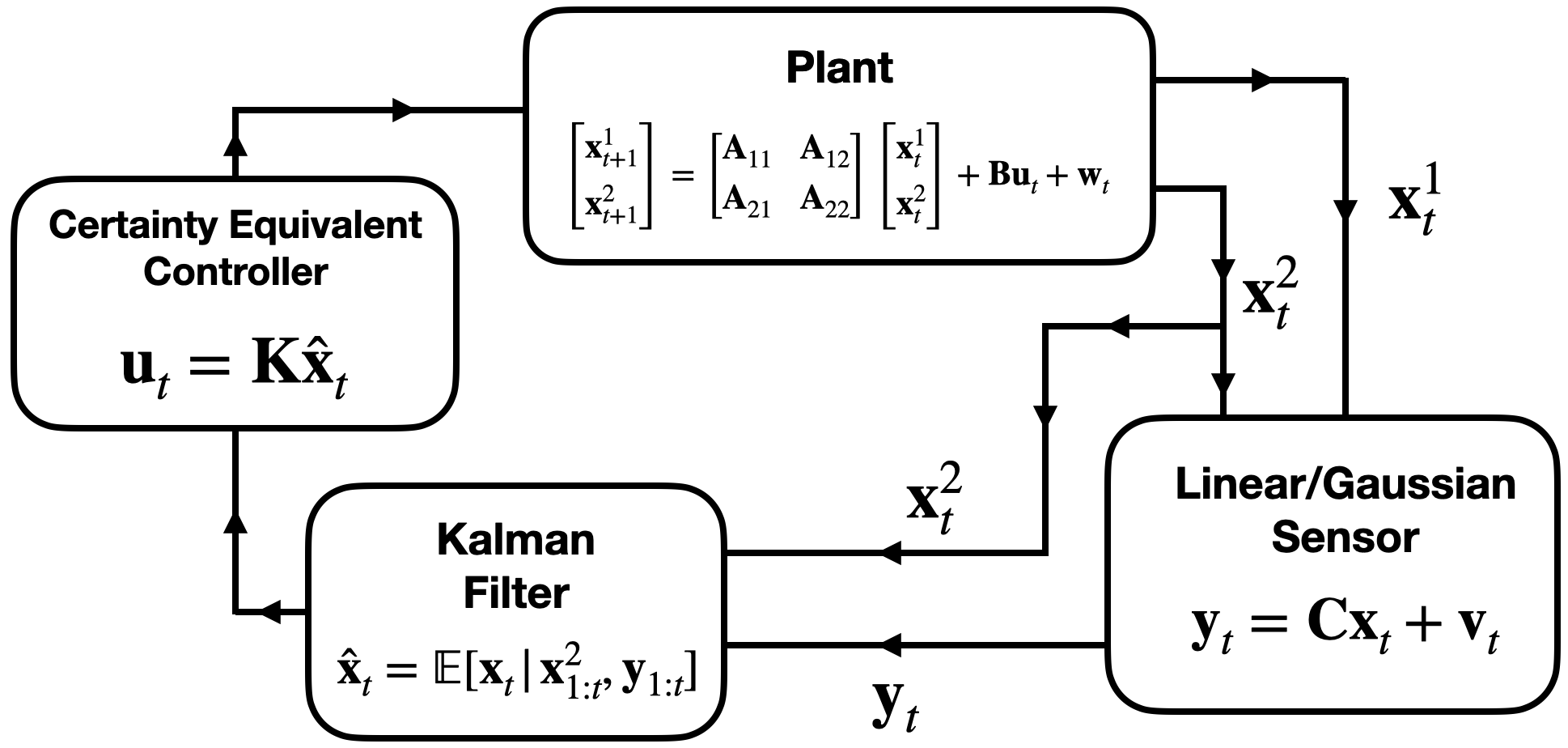}
    \vspace{-.4cm}
	\caption{The three-stage separation architecture.}
	\label{fig:conjecture}
\end{figure} 
\par{  I. \textbf{ Time-invariant linear/Gaussian sensor:} }
Let $\mathbf{C}_{1}\in \mathbb{R}^{n\times n}$ and $\mathbf{C}_{2}\in \mathbb{R}^{n\times m}$.The equation governing the sensor output, $\mathbf{y}_t$, is assumed to be
\begin{align}\label{eq:sensorModel}
    \mathbf{y}_t =  \begin{bmatrix}  \mathbf{C}_{1} &  \mathbf{C}_{2} \end{bmatrix}\mathbf{x}_t+\mathbf{v}_t,\text{ where }\mathbf{v}_t\sim\mathcal{N}(\mathbf{0},\mathbf{V}). \end{align}
\par{  II. \textbf{ Kalman filter:} }The standard Kalman filter (KF) computes the linear minimum mean squared error (LMMSE) estimator, which in the joint Gaussian case is also the MMSE estimator. The estimator is computed by the standard recursion  (cf. \cite{kailathBk}). The KF computes the estimate $\mathbf{\hat{x}}_{t}$ via a linear (in all arguments), time varying, $\mathbf{C}$ and $\mathbf{V}$ dependent recursion denoted $\mathbf{\hat{x}}_{t} = \Psi_t( \mathbf{\hat{x}}_{t-1}, \mathbf{y}_{t},  \mathbf{x}^2_{t},\mathbf{u}_{t})$. 
\par{III. \textbf{ Certainty equivalence control:} }We assume certainty equivalence linear feedback control. Let $\mathbf{S}$ be a stabilizing solution to the algebraic Riccati equation\cite{tanakaISIT}
    \begin{align}
       \mathbf{S} = \mathbf{A}^{\mathrm{T}}\mathbf{S}\mathbf{A}-\mathbf{A}^{\mathrm{T}}\mathbf{S}\mathbf{B}(\mathbf{B}^{\mathrm{T}}\mathbf{S}\mathbf{B}+\mathbf{R})^{-1}\mathbf{B}^{\mathrm{T}}\mathbf{S}\mathbf{A}+\mathbf{Q}.
    \end{align} The feedback control gain $\mathbf{K}$ is then given by
    \begin{align}\label{eq:kdef}
        \mathbf{K} = -(\mathbf{B}^{\mathrm{T}}\mathbf{S}\mathbf{B}+\mathbf{R})^{-1}\mathbf{B}^{\mathrm{T}}\mathbf{S}\mathbf{A}.
    \end{align} Under the three-stage test channel assumption, the design variables are limited to $\mathbf{C}$ and $\mathbf{V}\succeq 0$, converting (\ref{eq:desired_continuous}) into a finite-dimensional optimization. 
\section{A convex programming approach to the rate/control performance tradeoff}\label{sec:cvx}
Via three-stage separation, the minimum in  (\ref{eq:desired_continuous}) is given by 
\begin{subequations}\label{eq:almost_tractable}
\begin{align}
\underset{\mathbf{C}, \mathbf{V}}{\inf}   \quad & \underset{T\rightarrow \infty}{\lim\sup}\text{ } \frac{1}{T}I(\mathbf{x}_{1:T}\rightarrow\mathbf{y}_{1:T}|| \mathbf{x}^2_{1:T})  \label{eq:communicationCost} \\
\quad \text{s.t. }\forall\text{ }t &\quad  \underset{T\rightarrow \infty}{\lim\sup}\text{ }\frac{1}{T}\sum\nolimits_{t=1}^{T}\mathbb{E}[\lVert \mathbf{x}_{t+1} \rVert_{\mathbf{Q}}^{2}+\lVert \mathbf{u}_{t} \rVert_{\mathbf{R}}^{2}] \le \gamma, \label{eq:controlCost} 
\\ \quad &\quad  \mathbf{x}_{t+1} = \mathbf{A}\mathbf{x}_{t}+\mathbf{B}\mathbf{K}\mathbf{u}_{t}+\mathbf{w}_{t},\nonumber \\ \quad &\quad 
 \mathbf{y}_{t}=\mathbf{C}\mathbf{x}_{t}+\mathbf{v}_{t}\text{,    }\mathbf{v}_{t}\sim\mathcal{N}(\mathbf{0},\mathbf{V}),\text{ }\mathbf{V}\succeq 0,
\nonumber \\ \quad &\quad \mathbf{\hat{x}}_t = \Psi_t( \mathbf{\hat{x}}_{t-1}, \mathbf{y}_{t},  \mathbf{x}^2_{t}, \mathbf{u}_{t-1}) \text{, }\mathbf{u}_t = \mathbf{K}\mathbf{\hat{x}}_t, \nonumber
\end{align}
\end{subequations}
where we identify the DI  (\ref{eq:communicationCost}) as the \textit{communication cost} and the quadratic  (\ref{eq:controlCost}) as the \textit{control cost}\cite{oronNew}. All expectations are under the measure induced by $\mathbf{C}$, $\mathbf{V}$, and Fig. \ref{fig:conjecture}. In this section we derive a convex program from (\ref{eq:almost_tractable}). We first simplify the cost (\ref{eq:communicationCost}) under the assumed architecture, deriving an expression in terms of Kalman filter error covariance matrices. 
\subsection{The rate and control costs in terms of KF variables}\label{ss:simpCost}
Under the architecture in Fig. \ref{fig:conjecture}, it can be verified that 
$\mathbf{x}_{1:t-1}\leftrightarrow \mathbf{x}^2_{1:t}, \mathbf{x}^1_t, \mathbf{y}_{1:t-1}\leftrightarrow \mathbf{y}_t$. Thus, by the chain rule $I(\mathbf{x}_{1:t};\mathbf{y}_{t} | \mathbf{x}^2_{1:t}, \mathbf{y}_{1:t-1}) =    I(\mathbf{x}_{t};\mathbf{y}_{t} | \mathbf{x}^2_{1:t}, \mathbf{y}_{1:t-1})$ and the communication cost (\ref{eq:communicationCost}) is given by
\begin{multline}\label{eq:mi_x_y}
    I(\mathbf{x}_{1:T}\rightarrow\mathbf{y}_{1:T}|| \mathbf{x}^2_{1:T}) = \sum\nolimits_{t=1}^{T}I(\mathbf{x}_{t};\mathbf{y}_{t} | \mathbf{x}^2_{1:t}, \mathbf{y}_{1:t-1}).
\end{multline}  Let $\mathbf{\tilde{x}}^1_{t}=\mathbb{E}[\mathbf{x}^{1}_{t}|\mathbf{x}^2_{1:t}, \mathbf{y}_{1:t-1}]$ and $\mathbf{\hat{x}}^1_{t} = \mathbb{E}[\mathbf{x}^{1}_{t}|\mathbf{x}^2_{1:t}, \mathbf{y}_{1:t}]$. Denote the residuals
$\mathbf{\tilde{r}}_{t}= \mathbf{x}^1_{t}-\mathbf{\tilde{x}}^1_{t}$ and $\mathbf{\hat{r}}_{t}=\mathbf{{x}}^1_{t}-\mathbf{\hat{x}}^1_{t}$. Since $\mathbf{\tilde{x}}^1_{t}$ and $\mathbf{\hat{x}}^1_{t}$  are measurable functions of $\mathbf{x}^2_{1:t}, \mathbf{y}_{1:t-1}$ and $\mathbf{x}^2_{1:t}, \mathbf{y}_{1:t}$ respectively,  by the definition of MI
\begin{multline}\label{eq:main_cf_entropies}
    I(\mathbf{x}^1_{t};\mathbf{y}_{t} | \mathbf{x}^2_{1:t}, \mathbf{y}_{1:t-1}) =\\ h( \mathbf{\tilde{r}}_{t} | \mathbf{x}^2_{1:t}, \mathbf{y}_{1:t-1})-h( \mathbf{\hat{r}}_{t} | \mathbf{x}^2_{1:t}, \mathbf{y}_{1:t}). 
\end{multline} By the joint Gaussianity of $\mathbf{x}_{1:t}$ and $\mathbf{y}_{1:t}$ and the orthogonality principle, $\mathbf{\tilde{r}}_{t}$ is Gaussian, has $\mathbb{E}[\mathbf{\tilde{r}}_{t}]=\mathbf{0}$, and is independent of $\mathbf{x}^2_{1:t}, \mathbf{y}_{1:t-1}$. Likewise $\mathbf{\hat{r}}_{t}$ is Gaussian, has $\mathbb{E}[\mathbf{\hat{r}}_{t}]=\mathbf{0}$, and is independent of $\mathbf{x}^2_{1:t}, \mathbf{y}_{1:t}$. Define $\mathbf{\tilde{P}}_{t} = \mathbb{E}[\mathbf{\tilde{r}}_{t}\mathbf{\tilde{r}}_{t}^{\mathrm{T}}]$ and $\mathbf{\hat{P}}_{t} = \mathbb{E}[\mathbf{\hat{r}}_{t}\mathbf{\hat{r}}_{t}^{\mathrm{T}}]$. The differential entropy of $\mathbf{z}\sim\mathcal{N}(\mathbf{0}_{d},\boldsymbol{\Sigma})$ is $h(\mathbf{z}) = \frac{1}{2}(\log\det(\boldsymbol{\Sigma})+d\log(2\pi e))$ \cite{elemIT}. Thus (\ref{eq:main_cf_entropies}) is  \begin{align}
      I(\mathbf{x}^1_{t};\mathbf{y}_{t} | \mathbf{x}^2_{1:t}, \mathbf{y}_{1:t-1}) =  \frac{1}{2}(\log\det\mathbf{\tilde{P}}_{t}-\log\det\mathbf{\hat{P}}_{t}).
\end{align} Thus, the rate cost function in (\ref{eq:almost_tractable}) may be written
\begin{multline}\label{eq:rateCostCov}
   \underset{t\rightarrow\infty}{\lim\sup}\text{ } \frac{1}{T}I(\mathbf{x}_{1:T}\rightarrow\mathbf{y}_{1:T}|| \mathbf{x}^2_{1:T})  = \\ \underset{t\rightarrow\infty}{\lim\sup}\frac{1}{2T}\sum\nolimits_{t=1}^{T}\log\det\mathbf{\tilde{P}}_{t}-\log\det\mathbf{\hat{P}}_{t}.
\end{multline} Under the present assumptions (cf. \cite{tanakaISIT}\cite{SDP_DI}), the control cost may also be written in terms of $\mathbf{\hat{P}}_{t}$. Let $\boldsymbol{\Theta}$ be the upper left $n\times n$ block of $\mathbf{K}^{\mathrm{T}}(\mathbf{B}^{\mathrm{T}}\mathbf{S}\mathbf{B}+\mathbf{R})\mathbf{K}$. We have
\begin{multline}\label{eq:controlCostCov}
 \underset{T\rightarrow \infty}{\lim\sup}\frac{1}{T}\sum\nolimits_{t=1}^{T}\mathbb{E}[\lVert \mathbf{x}_{t+1} \rVert_{\mathbf{Q}}^{2} +\lVert \mathbf{u}_{t} \rVert_{\mathbf{R}}^{2}] =\\ \underset{T\rightarrow \infty}{\lim\sup}\frac{1}{T}\sum\nolimits_{i=1}^{\infty}\text{Tr}(\boldsymbol{\Theta}\mathbf{\hat{P}}_{t})+\text{Tr}(\mathbf{S}\mathbf{W}).
\end{multline} In the sequel, we recast (\ref{eq:almost_tractable}) in terms of $\mathbf{\hat{P}}_{t}$ and $\mathbf{\tilde{P}}_{t}$.
\subsection{The constraints in terms of Kalman filter variables}\label{ssec:constKalman}
In this subsection, we derive constraints between the residual covariance matrices and conclude the simplification of (\ref{eq:almost_tractable}). The sequences $\{\mathbf{\tilde{P}}_{t}\}$ and $\{\mathbf{\hat{P}}_{t}\}$ are related via a Riccati recursion we derive via considering the implementation of the Kalman filter from Fig. \ref{fig:conjecture} depicted in Fig. \ref{fig:kf2}.  \begin{figure}[]
	\centering
	\includegraphics[scale = .17]{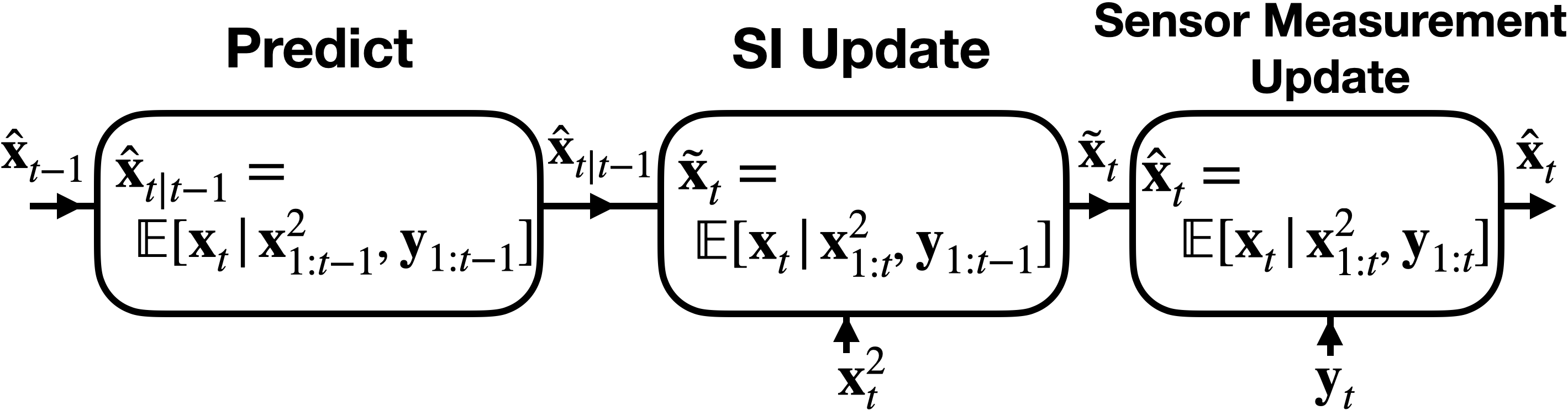}
	\vspace{-.45cm}
	\caption{A depiction of the Kalman filtering process with two measurement updates. The first update is after acquiring the SI $(\mathbf{x}^2)_t$, meanwhile the second is after acquiring the sensor measurement $\mathbf{y}_t$. In the present setting, joint Gaussianity ensures the filter computes MMSE estimators. The residuals are uncorrelated, and thus independent, of the respective observations (cf. \ref{ss:simpCost}).}
 	\label{fig:kf2}
\end{figure}  %idk why I need this indent here

Define the \textit{a posteriori} state estimate of $\mathbf{x}$ at time $t-1$ as $\mathbf{\hat{x}}_{t-1}$. This is the estimator given $\mathbf{x}^2_{1:t-1}$ and $\mathbf{y}_{1:t-1}$ and is given by 
    $\mathbf{\hat{x}}_{t-1} = [(\mathbf{\hat{x}}^1_{t-1})^{\mathrm{T}}\text{, }  (\mathbf{x}_{t-1}^2)^{\mathrm{T}}]^{\mathrm{T}}$ (cf. Sec. \ref{sec:cvx}). 
Since $\mathbf{x}_{t-1}^2$ is observed noiselessly there is no error in estimating $\mathbf{x}_{t-1}^2$; we thus defined the residual, $\mathbf{\hat{r}}_{t-1}$, with respect to $\mathbf{x}^1_{t-1}$ (only).  The orthogonality principle and Gaussianity ensures that $\mathbf{\hat{r}}_{t-1}$ is independent of $\mathbf{x}^2_{1:t-1},\mathbf{y}_{1:t-1}$. 

Denote the \textit{a priori} state estimate for time $t$ as $\mathbf{\hat{x}}_{t|t-1}$. Given the linear feedback control, $\mathbf{\hat{x}}_{t|t-1}$ is a linear function of  $\mathbf{\hat{x}}_{t-1}$ and is precisely the MMSE estimator $\mathbf{\hat{x}}_{t|t-1} = \mathbb{E}[\mathbf{x}_{t-1}|\mathbf{x}^2_{1:t-1}, \mathbf{y}_{1:t-1}]$. Denote the a priori residual process $\mathbf{\hat{r}}_{t|t-1} =  \mathbf{x}_{t}-\mathbf{\hat{x}}_{t|t-1}$.
 Note that in contrast to the definition of $\mathbf{\hat{r}}_{t-1}$,  $\mathbf{\hat{r}}_{t|t-1}$ contains residuals from estimating (predicting) both $\mathbf{x}^1_{t}$ and
$\mathbf{x}^2_{t}$.  It can be shown that $\mathbb{E}[ \mathbf{\hat{r}}_{t|t-1} ] = \mathbf{0}$. Denote the covariance matrix 
$\mathbf{P}_{t|t-1} = \mathbb{E}[\mathbf{\hat{r}}_{t|t-1}\mathbf{\hat{r}}_{t|t-1}^{\mathrm{T}}]$.  Let $\mathbf{\bar{A}} = [\mathbf{A}_{11}^{\mathrm{T}},\mathbf{A}_{21}^{\mathrm{T}}
    ]^{\mathrm{T}}$. 
By direct substitution $\mathbf{P}_{t|t-1} =\mathbf{\bar{A}}\mathbf{\hat{P}}_{t-1}\mathbf{\bar{A}}^{\mathrm{T}}+\mathbf{W}$, where $\mathbf{\hat{P}}_{t-1}$ is covariance of $\mathbf{\hat{r}}_{t-1}$ defined in \ref{ss:simpCost}. 

The estimator after the SI update (the noiseless observation of $\mathbf{x}^2_{t}$)  at time $t$ is given by $ \mathbf{\tilde{x}}_{t} = [(\mathbf{\tilde{x}}^1_{t})^{\mathrm{T}}\text{, } (\mathbf{x}_{t}^2)^{\mathrm{T}} ]^{\mathrm{T}}$ (cf. Sec. \ref{sec:cvx}). Again, $\mathbf{\tilde{x}}_{t} =\mathbb{E}[\mathbf{x}_t|\mathbf{x}^2_{1:t},\mathbf{y}_{1:t-1}]$ and is a linear function of $\mathbf{\hat{x}}_{t|t-1}$ and $\mathbf{x}^2_{t}$. The residual, $\mathbf{\tilde{r}}_{t}$, is again defined with respect to the error estimating $\mathbf{x}^1_{t}$ only (as in Sec. \ref{sec:cvx}). Let $\mathbf{P}^{11}_{t|t-1}\in \mathbb{R}^{n\times n}$, $\mathbf{P}^{12}_{t|t-1}\in \mathbb{R}^{n\times m}$, $\mathbf{P}^{21}_{t|t-1}\in \mathbb{R}^{m\times n}$, $\mathbf{P}^{22}_{t|t-1}\in \mathbb{R}^{m\times m}$ be such that
\begin{align}
    \mathbf{P}_{t|t-1} = \begin{bmatrix} \mathbf{P}_{t|t-1}^{11} & \mathbf{P}_{t|t-1}^{12} \\ \mathbf{P}_{t|t-1}^{21} & \mathbf{P}_{t|t-1}^{22}
    \end{bmatrix}.
\end{align} The covariance of the residual $\mathbf{\tilde{r}}_{t}$ (cf. \ref{ss:simpCost}) follows from a standard Shur complement result
\begin{align}\label{eq:sideInfoUpdate}
    \mathbf{\tilde{P}}_{t} = \mathbf{P}^{11}_{t|t-1} - \mathbf{P}^{12}_{t|t-1}(\mathbf{P}^{22}_{t|t-1})^{-1}\mathbf{P}^{21}_{t|t-1}.
\end{align} Finally, the sensor measurement update computes the posterior state estimate at time $t$. It can be shown that $\mathbf{\hat{P}}_{t}$ is given by   \begin{align}\label{eq:informationUpdateMeasurement}
    \mathbf{\hat{P}}_{t}^{-1}=\mathbf{\tilde{P}}_t^{-1}+ \mathbf{C}_{1}^{\mathrm{T}}\mathbf{V}^{-1}\mathbf{C}_{1},
\end{align}  which demonstrates that $\mathbf{C}_{2}$ is completely arbitrary.

Let $\mathbf{F}=\mathbf{A}_{21}^{\mathrm{T}}\mathbf{W}_{22}^{-1}\mathbf{A}_{21}$. Using (\ref{eq:sideInfoUpdate}), the matrix inversion lemma gives
\begin{align}\label{eq:priorNewNotation}
    \mathbf{\tilde{P}}_{t+1} = \mathbf{W}_{11}+\mathbf{A}_{11}\left(\mathbf{\hat{P}}_t^{-1}+\mathbf{F}\right)^{-1}\mathbf{A}_{11}^{\mathrm{T}}.
\end{align} Substituting  (\ref{eq:informationUpdateMeasurement}) into (\ref{eq:priorNewNotation}) gives a recursion for $\mathbf{\tilde{P}}$ via
\begin{align}\label{eq:RDE}
    \mathbf{\tilde{P}}_{t+1} = \mathbf{W}_{11}+ \mathbf{A}_{11}\left(\mathbf{\tilde{P}}^{-1}_{t}+\mathbf{C}_{1}^{\mathrm{T}}\mathbf{V}^{-1}\mathbf{C}_{1}+\mathbf{F} \right)^{-1}\mathbf{A}_{11}^{\mathrm{T}}.
\end{align} The matrix inversion lemma demonstrates that (\ref{eq:RDE}) is a Riccati difference equation \cite{RDE_convergence}. Given an initial condition, the recursion (\ref{eq:RDE}) converges under a variety of circumstances\cite{RDE_convergence}\cite{kailathBk}. If it exists, the steady state solution $\mathbf{\tilde{P}}_{\infty}$ solves the discrete algebraic Riccati equation
\begin{align}\label{eq:DARE}
     \mathbf{\tilde{P}}_{\infty} = \mathbf{W}_{11}+ \mathbf{A}_{11}\left(\mathbf{\tilde{P}}^{-1}_{\infty}+\mathbf{C}_{1}^{\mathrm{T}}\mathbf{V}^{-1}\mathbf{C}_{1}+\mathbf{F} \right)^{-1}\mathbf{A}_{11}^{\mathrm{T}}.
\end{align}  In particular, \cite[Theorem 4.1]{RDE_convergence} establishes convergence to a unique, positive definite solution when $(\mathbf{A}_{11},\mathbf{W}^{\frac{1}{2}}_{11})$ is stabilizable and $([\mathbf{C}_{1}^{\mathrm{T}},\mathbf{A}_{21}^{\mathrm{T}}]^{\mathrm{T}},\mathbf{A}_{11})$ is detectable \cite{RDE_convergence}. The stabilizability is immediate as $\mathbf{W}_{11}>0$. Furthermore, in the present setting, the existence of a positive definite solution to (\ref{eq:DARE}) can be shown to imply that $([\mathbf{C}_{1}^{\mathrm{T}},\mathbf{A}_{21}^{\mathrm{T}}]^{\mathrm{T}},\mathbf{A})$ is detectable via a discrete time Liaponov equation.  We restrict our attention to the case that  $([\mathbf{C}_{1}^{\mathrm{T}},\mathbf{A}_{21}^{\mathrm{T}}]^{\mathrm{T}},\mathbf{A})$ is detectable.

Convergence of $\{\mathbf{\tilde{P}}_{t}\}$ implies that  $\{\mathbf{\hat{P}}_{t}\}$ also converges. The limits of $\{\mathbf{\hat{P}}_{t}\}$ and $\{\mathbf{\tilde{P}}_{t}\}$  must satisfy both
\begin{subequations}\label{eq:convergenceConstraints}
\begin{align}\label{eq:tilde2hat}
\mathbf{\hat{P}}^{-1}_{\infty} = \mathbf{\tilde{P}}_{\infty}^{-1}+\mathbf{C}_{1}^{\mathrm{T}}\mathbf{V}^{-1}\mathbf{C}_{1}&\text{, and}
\end{align}
\begin{align}\label{eq:hat2tilde}
\mathbf{\tilde{P}}_{\infty} = \mathbf{W}_{11}+\mathbf{A}_{11}\left(\mathbf{\hat{P}}_\infty^{-1}+\mathbf{F}\right)^{-1}\mathbf{A}_{11}^{\mathrm{T}}.
\end{align} 

\end{subequations} Given $\mathbf{C}$ and $\mathbf{V}$, if such a $\mathbf{\tilde{P}}_{\infty}>0$ and $\mathbf{\hat{P}}_{\infty}>0$ can be found, the resulting $\mathbf{\tilde{P}}_{\infty}$ will satisfy (\ref{eq:DARE}). If for some $\mathbf{C}$ and $\mathbf{V}$ there exists $\mathbf{\tilde{P}}_{\infty}$ (necessarily positive definite) satisfying (\ref{eq:convergenceConstraints}), it follows that $([\mathbf{C}_{1}^{\mathrm{T}},\mathbf{A}_{21}^{\mathrm{T}}]^{\mathrm{T}},\mathbf{A})$ is detectable and that both $\mathbf{\tilde{P}}_{t}\rightarrow\mathbf{\tilde{P}}_{\infty}$ and $\mathbf{\hat{P}}_{t}\rightarrow\mathbf{\hat{P}}_{\infty}$. Using (\ref{eq:rateCostCov}), a standard Ces\'{a}ro mean (cf. \cite{elemIT}) argument gives that
\begin{align}\label{eq:rateSimp}
     \underset{T\rightarrow\infty}{\lim\sup}\text{ } \frac{1}{T}I(\mathbf{x}_{1:T}\rightarrow\mathbf{y}_{1:T}||\mathbf{x}^2_{1:T}) = \frac{1}{2}\log\frac{\det\mathbf{\tilde{P}}_{\infty}}{\det\mathbf{\hat{P}}_{\infty}}.
\end{align} Similarly, using (\ref{eq:controlCostCov}) we have that 
\begin{align}\label{eq:controlSimp}
 \underset{T\rightarrow \infty}{\lim\sup}\frac{1}{T}\sum\limits_{t=1}^{T}\mathbb{E}[\lVert \mathbf{x}_{t+1} \rVert_{\mathbf{Q}}^{2} +\lVert \mathbf{u}_{t} \rVert_{\mathbf{R}}^{2}] = \text{Tr}(\mathbf{\Theta}\mathbf{\hat{P}}_{\infty}+\mathbf{S}\mathbf{W}).
\end{align} In the following subsection, we use these results to derive a convex program for the rate distortion problem (\ref{eq:almost_tractable}).

\subsection{Derivation of the convex program}\label{ssec:policyOpt}
Define $\mathbf{\hat{P}}\overset{\Delta}{=}\mathbf{\hat{P}}_{\infty}$ and $\mathbf{\tilde{P}}\overset{\Delta}{=}\mathbf{\tilde{P}}_{\infty}$. Substituting  (\ref{eq:convergenceConstraints}), (\ref{eq:rateSimp}), and  (\ref{eq:controlSimp}) into  (\ref{eq:almost_tractable}) yields the finite dimensional optimization  
\begin{mini}|s|
{\mathbf{C}, \mathbf{V}}{\frac{1}{2}(\log\det\mathbf{\tilde{P}}-\log\det\mathbf{\hat{P}}) }
{\label{eq:finiteDimensional}}{}
\addConstraint{\mathbf{V}\succeq 0\text{, } \mathbf{\tilde{P}}\succeq 0\text{, } \text{Tr}(\mathbf{\Theta}\mathbf{\hat{P}}+\mathbf{S}\mathbf{W}) \le \gamma}
\addConstraint{\mathbf{\hat{P}}^{-1} = \mathbf{\tilde{P}}^{-1}+\mathbf{C}_{1}^{\mathrm{T}}\mathbf{V}^{-1}\mathbf{C}_{1}}{}
\addConstraint{\mathbf{\tilde{P}} = \mathbf{W}_{11}+\mathbf{A}_{11}(\mathbf{\hat{P}}^{-1}+\mathbf{F})^{-1}\mathbf{A}_{11}^{\mathrm{T}} }.
\end{mini} Since $\mathbf{W}_{11}\succ 0$,  $\mathbf{\tilde{P}}\succ 0$. The minimum in (\ref{eq:finiteDimensional}) can be  found\goodbreak\noindent by the convex optimization
\begin{mini}|s|
{\mathbf{\hat{P}},\boldsymbol{\Pi}}{\frac{\log\det\mathbf{W}-\log\det\boldsymbol{\Pi}-\log\det{
( \mathbf{W}_{22}+\mathbf{A}_{21}\mathbf{\hat{P}}\mathbf{A}_{21}^{\mathrm{T}})}}{2} }
{\label{eq:convexProgram}}{}
\addConstraint{\mathbf{\hat{P}}\succ {0}, \boldsymbol{\Pi}\succeq 0, \text{Tr}(\mathbf{\Theta}\mathbf{\hat{P}})+\text{Tr}(\mathbf{S}\mathbf{W}) \le \gamma}
\addConstraint{  \mathbf{W}+\mathbf{\bar{A}}\mathbf{\hat{P}}\mathbf{\bar{A}}^{\mathrm{T}} -\begin{bmatrix}\mathbf{\hat{P}} & 0\\ 0 & 0\end{bmatrix} \succeq 0 }{}
\addConstraint{ \begin{bmatrix}
\mathbf{\hat{P}}-\boldsymbol{\Pi} & \mathbf{\hat{P}}\mathbf{\bar{A}}^{\mathrm{T}} \\ \mathbf{\bar{A}}\mathbf{\hat{P}} & \mathbf{W}+\mathbf{\bar{A}}\mathbf{\hat{P}}\mathbf{\bar{A}}^{\mathrm{T}}
\end{bmatrix}\succeq 0}.
\end{mini}  Details are given in \cite[Appendix \ref{app:convexification}]{our_appendix}. Let $\mathbf{\hat{P}}_{\min}$ be the minimizer in (\ref{eq:convexProgram}), and let  $\mathbf{\tilde{P}}_{\min}$ be given by (\ref{eq:hat2tilde}). The minimizers $\mathbf{C}_{1}$ and $\mathbf{V}$ are the set of matrices satisfying $\mathbf{\hat{P}}_{\min}^{-1}-\mathbf{\tilde{P}}^{-1}_{\min}=\mathbf{C}_{1}\mathbf{V}^{-1}\mathbf{C}^{\mathrm{T}}_{1}$. Without loss of generality, we choose $\mathbf{V}=\mathbf{I}$, $\mathbf{C}_{1}$ the corresponding minimizer, and $\mathbf{C}_{2}=\mathbf{0}$. We now show that the minimum is nearly achievable in the architecture of Fig. \ref{fig:flowchart}. 
\section{Quantization and Prefix Free Coding}\label{sec:achieve}
The architecture used to demonstrate the achievability  result follows from \cite[IV]{tanakaISIT}, and is shown in Fig. \ref{fig:quantizerImp}. As in \cite{tanakaISIT}, we use  a predictive elementwise uniform quantizer with subtractive dither.  
\begin{figure}[]
	\centering
	\includegraphics[scale = .215]{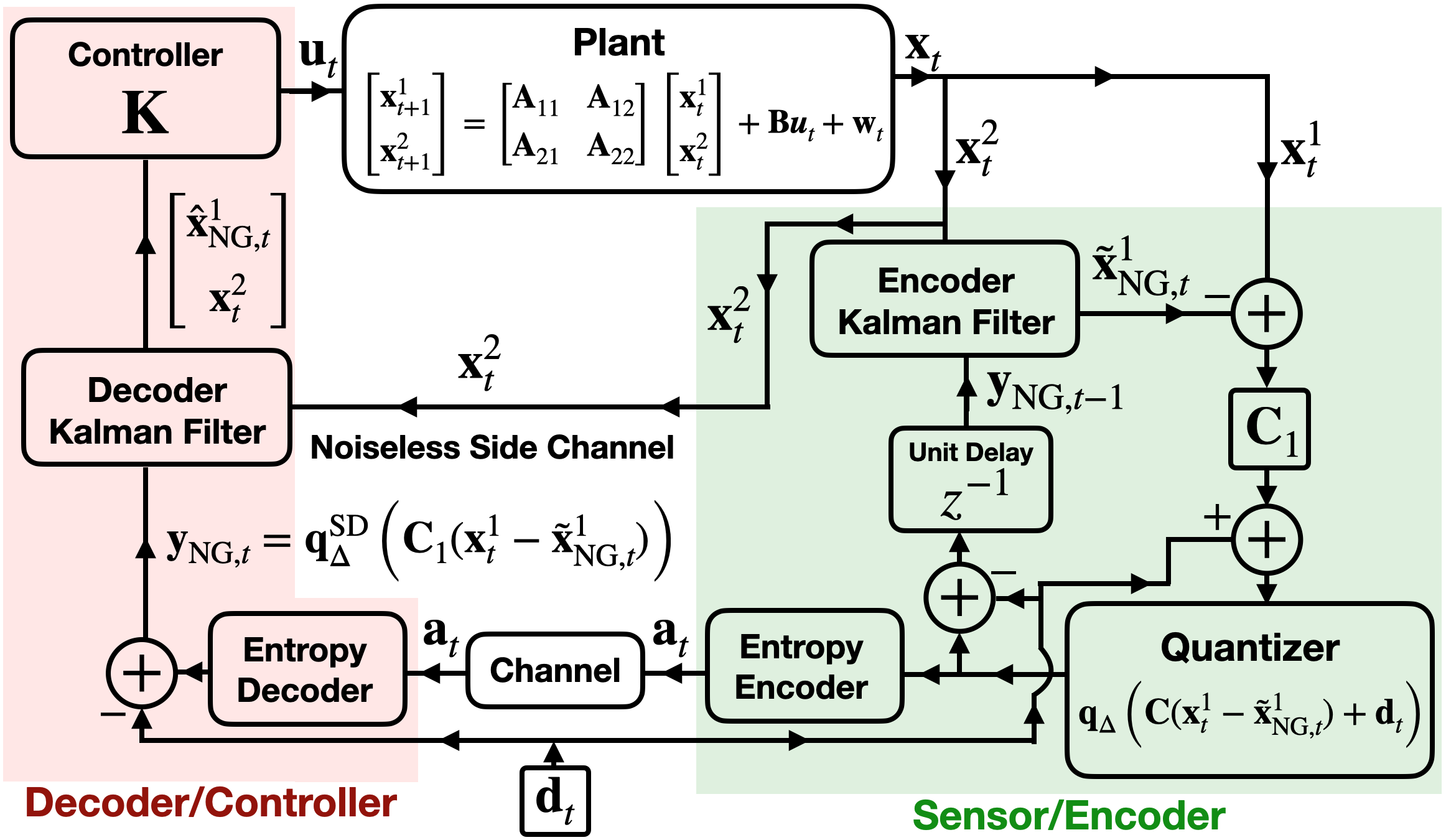}
    \vspace{-.3cm}
	\caption{The dither signal $[\mathbf{d}_{t}]_{i}\sim \text{Uniform}([-\frac{\Delta}{2},\frac{\Delta}{2}))$ IID over $i$, $t$ is independent of $\mathbf{x}_{1:t}$, $\mathbf{y}_{1:t-1}$, $\mathbf{u}_{1:t-1}$, $\mathbf{a}_{1:t-1}$ but is assumed to be known at both the encoder and decoder. In practice, this ``shared randomness" could be accomplished by using synchronized pseudorandom number generators at both the encoder and decoder. }
	\label{fig:quantizerImp}
\end{figure} 
We define an elementwise uniform quantizer with sensitivity $\Delta$ as a function $\mathbf{q}_{\Delta}: \mathbb{R}^{n}\rightarrow\mathbb{R}^n$ such that
\begin{align}
    [\mathbf{q}_{\Delta}(\mathbf{z})]_{i} = m\Delta \text{ if } [\mathbf{z}]_{i}\in [m\Delta-\frac{\Delta }{2},m\Delta+\frac{\Delta}{2}),
\end{align} e.g. each element of $\mathbf{z}$ is ``rounded" to the nearest integer multiple of $\Delta$. For a random input $\mathbf{z}$, $\mathbf{q}_{\Delta}(\mathbf{z})$ is a discrete RV with countable support. Consider the random vector $\mathbf{d}\in\mathbb{R}^{n}$ where $[\mathbf{d}]_{i}\sim\text{Uniform}[-\frac{\Delta}{2},\frac{\Delta}{2}]$ IID over $i$ and independent of $\mathbf{z}$. Define the quantizer with \textit{subtractive dither} via
\begin{align}\label{eq:dither}
    \mathbf{q}^{\mathrm{SD}}_{\Delta}(\mathbf{z}) = \mathbf{q}_{\Delta}(\mathbf{z}+\mathbf{d})-\mathbf{d} 
\end{align} Dithering allows the quantization error to manifest as additive \textit{uniform} noise; it can be shown that $\mathbf{n}=\mathbf{z}- \mathbf{q}^{\mathrm{SD}}_{\Delta}(\mathbf{z})$ is independent of $\mathbf{z}$ and that the elements $[\mathbf{n}]_{i}$ are IID with $[\mathbf{n}]_{i}\sim\text{Uniform}[-\frac{\Delta}{2},\frac{\Delta}{2}]$\cite[Lemma 1a]{tanakaISIT}\cite{ditherQuant}.
The caption of Fig. \ref{fig:quantizerImp} outlines the use of dithering this achievability result.\goodbreak\noindent

We now show that when $\Delta = 2\sqrt{3}$, the system in Fig. \ref{fig:quantizerImp} achieves an equivalent control performance as the as the architecture in Fig. \ref{fig:conjecture} for equivalent $\mathbf{C}_1$ and $\mathbf{V}=\mathbf{I}$. In Fig. \ref{fig:quantizerImp}, at time $t$ the decoder observes a dithered quantized measurement of $\mathbf{x}^1$, denoted $\mathbf{y}^{\mathrm{NG}}_{t}$ and to be described presently. The measurement is predictive and defined recursively via an encoder KF process. At time $t$, a KF at the encoder computes
\begin{align}
\mathbf{\tilde{x}}^{1,\mathrm{NG}}_{t} = \text{The LMMSE estimate of $\mathbf{x}^1$ given $\mathbf{y}^{\mathrm{NG}}_{1:t-1}$, $\mathbf{x}^2_{1:t}$}. \nonumber
\end{align} The encoder's quantizer computes the discrete $\mathbf{\tilde{z}}_t= \mathbf{q}_{\Delta}( \mathbf{C}_{1}(\mathbf{x}^1_{t}-\mathbf{\tilde{x}}^{1,\mathrm{NG}}_{t})+\mathbf{d}_{t})$, and encodes $\mathbf{\tilde{z}}_t$ with a prefix-free lossless Shannon-Fano-Elias (SFE) code. The codeword is sent to the decoder, which (exactly) reconstructs $\mathbf{\tilde{z}}_t$. 

Given the dither signal, the decoder forms $\mathbf{y}^{\mathrm{NG}}_{t}=\mathbf{\tilde{z}}_t-\mathbf{d}_{t}$, or, equivalently $\mathbf{y}^{\mathrm{NG}}_{t} = \mathbf{q}_{\Delta}^{\mathrm{SD}}\left( \mathbf{C}_{1}(\mathbf{x}^1_{t}-\mathbf{\tilde{x}}^{1,\mathrm{NG}}_{t})\right)$. This gives 
\begin{align}
    \mathbf{y}^{\mathrm{NG}}_{t} = \mathbf{C}_{1}\mathbf{x}^1_{t}-\mathbf{C}_{1}\mathbf{\tilde{x}}^{1,\mathrm{NG}}_{t}+\mathbf{n}_{t}\nonumber
\end{align} where $\mathbf{n}_{t}$ is a zero mean, uniform random vector with IID elements and $\mathbb{E}[\mathbf{n}_{t}\mathbf{n}_{t}^{\mathrm{T}}] = \mathbf{I}$. The decoder side Kalman filter operates analogously to the two stage filter in Fig. \ref{fig:kf2}. Having received the previous measurements $\mathbf{y}^{\mathrm{NG}}_{1:t-1}$ and the SI $\mathbf{x}^2_{1:t}$, the decoder can compute  $\mathbf{\tilde{x}}^{1,\mathrm{NG}}_{t}$ and form a centered measurement $\mathbf{\overline{y}}^{\mathrm{NG}}_{t}=\mathbf{y}^{\mathrm{NG}}_{t} + \mathbf{C}_{1}\mathbf{\tilde{x}}^{1,\mathrm{NG}}_{t}$. It clear that,
\begin{align}
    \mathbf{\hat{x}}^{1,\mathrm{NG}}_{t} = \text{The LMMSE estimate of $\mathbf{x}_{t}^1$ given $\mathbf{y}^{\mathrm{NG}}_{1:t}$,  $\mathbf{x}^2_{1:t}$},
\end{align} is the same as the LMMSE of $\mathbf{x}_{t}^1$ given $\mathbf{\overline{y}}^{\mathrm{NG}}_{1:t}$ and $\mathbf{x}^2_{1:t}$. The controller forms the control input $\mathbf{u}_{t}= \mathbf{K}[
 (\mathbf{\hat{x}}^{1,\mathrm{NG}}_{t})^{\mathrm{T}} \text{, }(\mathbf{{x}}_{t}^2)^{\mathrm{T}}
]^{\mathrm{T}}$ where $\mathbf{K}$ is as in (\ref{eq:kdef}). A corollary to the proof of \cite[Lemma 1a]{tanakaISIT} demonstrates that under this  (really any) feedback arrangement, the sequence of quantization noises $\{\mathbf{n}_{t}\}$ is temporally white, e.g. $\mathbf{E}[\mathbf{n}_{t}\mathbf{n}_{t'}^{\mathrm{T}}] = \mathbf{0}$ if $t\neq t'$.

This leads to a result analogous to \cite[Lemma 2]{tanakaISIT}. Having fixed $\mathbf{C}_{1}$ and $\mathbf{V}=\mathbf{I}$, denote the jointly Gaussian random variables $(\mathbf{x}^{}_{t},\mathbf{\tilde{x}}^{1}_{t},  \mathbf{\hat{x}}^{1}_{t})$ with the joint distribution induced by the architecture in Fig. \ref{fig:conjecture} by $(\mathbf{x}^{\mathrm{G}}_{t}, \mathbf{\tilde{x}}^{1,\mathrm{G}}_{t}, \mathbf{\hat{x}}^{1,\mathrm{G}}_{t})$. Likewise, denote the (generally non-Gaussian) RVs  $(\mathbf{x}^{}_{t}, \mathbf{\tilde{x}}^{1}_{t}, \mathbf{\hat{x}}^{1}_{t}, )$ with the joint distribution induced by the architecture in Fig. \ref{fig:quantizerImp} by 
$(\mathbf{x}^{\mathrm{NG}}_{t}, \mathbf{\tilde{x}}^{1,\mathrm{NG}}_{t}, \mathbf{\hat{x}}^{1,\mathrm{NG}}_{t})$. We have the following lemma.
\begin{lemma}\label{lemma:secondMoments}
If RVs describing the initial conditions $\mathbf{x}^{\mathrm{NG}}_{1}$ and  $\mathbf{x}^{\mathrm{G}}_{1}$ have identical first and second moments, then the processes  $\{(\mathbf{x}^{\mathrm{NG}}_{t}, \mathbf{\hat{x}}^{1,\mathrm{NG}}_{t}, \mathbf{\tilde{x}}^{1,\mathrm{NG}}_{t})\}$ and $\{(\mathbf{x}^{\mathrm{G}}_{t}, \mathbf{\tilde{x}}^{1,\mathrm{G}}_{t}, \mathbf{\hat{x}}^{1,\mathrm{G}}_{t})\}$ are equivalent up to second moments. Regardless of initial conditions $
\mathbb{E}[(\mathbf{x}_{t}^{1,\mathrm{NG}}-\mathbf{\hat{x}}_{t}^{1,\mathrm{NG}})(\mathbf{x}_{t}^{1,\mathrm{NG}}-\mathbf{\hat{x}}_{t}^{1,\mathrm{NG}})^{\mathrm{T}}]\rightarrow \mathbf{\hat{P}}$.
\end{lemma}\noindent This result follows from comparing the measurement model
\begin{align}
     \mathbf{\overline{y}}^{\mathrm{NG}}_{t} =\mathbf{C}_{1}\mathbf{x}^{1,\mathrm{NG}}_{t}+\mathbf{n}_{t}
 \end{align} to the  linear/Gaussian model (\ref{eq:sensorModel}) under the assumed choices of $\mathbf{C}_{1}$ and $\mathbf{V}=\mathbf{I}$. While the additive  white noise is uniform, rather than Gaussian, it has $\mathbb{E}[\mathbf{n}_{t}]=\mathbf{0}$ and $\mathbb{E}[\mathbf{n}_{t}\mathbf{n}_{t}^{\mathrm{T}}]=\mathbf{I}$. The first statement follows from an induction on $t$. The latter follows as the Riccati recursion relating the covariance matrices of the error processes $\mathbf{x}_{t}^{1,\mathrm{NG}}-\mathbf{\hat{x}}_{t}^{1,\mathrm{NG}}$ and $\mathbf{x}_{t}^{1,\mathrm{NG}}-\mathbf{\tilde{x}}_{t}^{1,\mathrm{NG}}$ is identical to that derived in \ref{ssec:constKalman}. The same control cost is achieved in both systems (cf. (\ref{eq:controlSimp})).
 
It remains to bound the codeword length.  Recall the discrete \goodbreak\noindent RV $\mathbf{\tilde{z}}_t$, and define $\mathbf{z}_{t} =\mathbf{C}_{1}(\mathbf{x}^{1, \mathrm{NG}}_{t}-\mathbf{\tilde{x}}^{1,\mathrm{NG}}_{t})$. At every time $t$, by the SFE construction (cf. \cite{elemIT}) there exists a lossless, prefix-free code that encodes $\mathbf{\tilde{z}}_{t}$ with an expected length  $\mathbb{E}[\ell(\mathbf{a}_{t})]\le H(\mathbf{\tilde{z}}_{t}|\mathbf{d}_{t} )+2$. Consider the joint Gaussian case and define $\mathbf{\tilde{r}}^{\mathrm{G}}_{t}$ as in Sec. \ref{sec:cvx}. The next lemma is proved in \cite[Appendix \ref{app:ditherLemma}]{our_appendix}. 
\begin{lemma}[~\cite{tanakaISIT}]\label{lem:ditherLemma}
At every time $t$, we have
\begin{align}\label{eq:lemDithEq}
   H(\mathbf{\tilde{z}}_{t}|\mathbf{d}_{t}) \le \frac{n}{2}\log_{2}\frac{4\pi e}{12} +I(\mathbf{C}_{1}\mathbf{\tilde{r}}^{\mathrm{G}}_{t};\mathbf{C}_{1}\mathbf{\tilde{r}}^{\mathrm{G}}_{t}+\mathbf{v}_{t}).
\end{align}
\end{lemma} Let $k = 2+ \frac{n}{2}\log_{2}\frac{4\pi e}{12}$. Our main result is the following. 
\begin{theorem}\label{th:MR}
When the entropy encoder and decoder in Fig. \ref{fig:quantizerImp} use SFE coding adapted to the PMF of $\mathbf{\tilde{z}}_{t}$ for all $t$, the architecture achieves \begin{multline}\label{eq:rdup}
     \underset{T\rightarrow \infty}{\lim\sup}\text{ }\frac{1}{T}\sum\limits_{t=1}^{T}\mathbb{E}[\ell(\mathbf{a}_{t})] \le \\ \underset{T\rightarrow \infty}{\lim\sup}\text{ }\frac{1}{T}  I(\mathbf{x}^{\mathrm{G}}_{1:T}\rightarrow\mathbf{y}^{\mathrm{G}}_{1:T}|| \mathbf{x}^{\mathrm{G},2}_{1:T}) + k.
 \end{multline}
\end{theorem}
\begin{proof}
At every time $t$, SFE codeword has a length $\mathbb{E}[\ell(\mathbf{a}_{t})]\le 1+H(\mathbf{\tilde{z}}_{t}|\mathbf{d}_{t})$.
Since $\mathbf{\tilde{x}}^{1,\mathrm{G}}_{t}$ is a measurable function of $\mathbf{x}^{\mathrm{G},2}_{1:t}, \mathbf{y}^{\mathrm{G}}_{1:t-1}$ we have that  \begin{multline}
     I(\mathbf{x}^{1,\mathrm{G}}_{t};\mathbf{y}^{\mathrm{G}}_{t} | \mathbf{x}^{2,\mathrm{G}}_{1:t}, \mathbf{y}^{\mathrm{G}}_{1:t-1})=\\I(\mathbf{\tilde{r}}^{\mathrm{G}}_{t};\mathbf{C}_{1}\mathbf{\tilde{r}}^{\mathrm{G}}_{t}+\mathbf{v}_{t}
     | \mathbf{x}^{2,\mathrm{G}}_{1:t}, \mathbf{y}^{\mathrm{G}}_{1:t-1}).
\end{multline}  Since  $\mathbf{\tilde{r}}^{\mathrm{G}}_{t}$ and $\mathbf{v}_{t}$ are independent of  $(\mathbf{x}^{2,\mathrm{G}}_{1:t}, \mathbf{y}^{\mathrm{G}}_{1:t-1})$, this implies 
\begin{align}
    I(\mathbf{x}^{1,\mathrm{G}}_{t};\mathbf{y}^{\mathrm{G}}_{t} | \mathbf{x}^{2,\mathrm{G}}_{1:t}, \mathbf{y}^{\mathrm{G}}_{1:t-1}) = I(\mathbf{\tilde{r}}^{\mathrm{G}}_{t};\mathbf{C}_{1}\mathbf{\tilde{r}}^{\mathrm{G}}_{t}+\mathbf{v}_{t})
\end{align} Note that both $\mathbf{\tilde{r}}^{\mathrm{G}}_{t}\leftrightarrow \mathbf{C}_{1}\mathbf{\tilde{r}}^{\mathrm{G}}_{t} \leftrightarrow \mathbf{C}_{1}\mathbf{\tilde{r}}^{\mathrm{G}}_{t}+\mathbf{v}_{t}$ and  $\mathbf{C}_{1}\mathbf{\tilde{r}}^{\mathrm{G}}_{t} \leftrightarrow \mathbf{\tilde{r}}^{\mathrm{G}}_{t}\leftrightarrow \mathbf{C}_{1}\mathbf{\tilde{r}}^{\mathrm{G}}_{t}+\mathbf{v}_{t}$ are Markov chains. Applying the standard data processing inequality (twice) to $I(\mathbf{\tilde{r}}^{\mathrm{G}}_{t};\mathbf{C}_{1}\mathbf{\tilde{r}}^{\mathrm{G}}_{t}+\mathbf{v}_{t})$ using both of these chains allows us to conclude that
\begin{align}
    I(\mathbf{x}^{1,\mathrm{G}}_{t};\mathbf{y}^{\mathrm{G}}_{t} | \mathbf{x}^{2,\mathrm{G}}_{1:t}, \mathbf{y}^{\mathrm{G}}_{1:t-1}) = I(\mathbf{C}_{1}\mathbf{\tilde{r}}^{\mathrm{G}}_{t};\mathbf{C}_{1}\mathbf{\tilde{r}}^{\mathrm{G}}_{t}+\mathbf{v}_{t}).
\end{align} Thus, by Lemma \ref{lem:ditherLemma} 
\begin{align}\label{eq:entropyBoundFinal}
H(\mathbf{\tilde{z}}_{t}|\mathbf{d}_{t}) \le \frac{n}{2}\log_{2}\frac{4\pi e}{12}+  I(\mathbf{x}^{1,\mathrm{G}}_{t};\mathbf{y}^{\mathrm{G}}_{t} | \mathbf{x}^{2,\mathrm{G}}_{1:t}, \mathbf{y}^{\mathrm{G}}_{1:t-1}).
\end{align} Summing (\ref{eq:lemDithEq}) over $t$, and applying  (\ref{eq:DIDEF}) gives (\ref{eq:rdup}).
\end{proof} The Ces\'{a}ro mean argument in (\ref{eq:rateSimp})  applied to (\ref{eq:rdup}) gives
\begin{align}
     \underset{T\rightarrow \infty}{\lim\sup}\text{ }\frac{1}{T}\sum\limits_{t=1}^{T}\mathbb{E}[\ell(\mathbf{a}_{t})] \le k+ \frac{\log\det\mathbf{\tilde{P}}_{\min}-\log\det\mathbf{\hat{P}}_{\min}}{2}\nonumber,
\end{align} which is convenient for computing the bound via (\ref{eq:convexProgram}). 
\bibliographystyle{IEEEtran}
\bibliography{IEEEabrv,control_info_theory}
\clearpage
\appendices
\section{A formal proof of the converse bound}\label{app:SIEO}
\begin{figure}[]
	\centering
	\includegraphics[scale = .31]{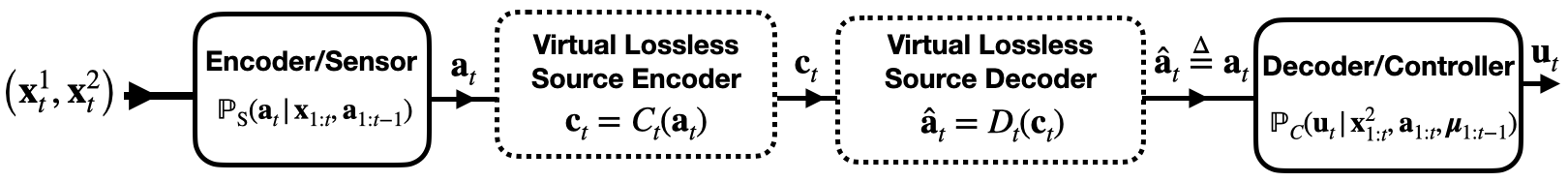}
    \vspace{-.3cm}
	\caption{A generalized version of the path from encoder to decoder in the model of Fig. \ref{fig:flowchart}. The minimum achievable rate for a system with the additional ``virtual" encoder/decoder pair lower bounds the minimum achievable rate for a system without the additional virtual pair.}
	\label{fig:sie}
\end{figure}
Consider the system model in Fig. \ref{fig:flowchart} at time $t$. By assumption (cf. Sec. \ref{sec:formulation}), $\mathbf{a}_{t}$ is a prefix-free codeword, so if $\mathsf{a}_{1}\neq\mathsf{a}_{2}$, $\mathbb{P}(\mathbf{a}_{t}=\mathsf{a}_{1})>0$ and $\mathbb{P}(\mathbf{a}_{t}=\mathsf{a}_{2})>0$ then $\mathsf{a}_{1}$ is not a prefix of $\mathsf{a}_{2}$ and vice-versa. The codeword $\mathbf{a}_{t}$ is a discrete variable with countable support chosen by the policy defined by kernel (a conditional PMF)
\begin{align}
\mathbb{P}_{S}(\mathbf{a}_{t}|\mathbf{x}^1_{1:t},\mathbf{x}^2_{1:t},\mathbf{a}_{1:t-1}),
\end{align} where we added the subscript $S$ to denote the ``encoder/sensor" policy as in Fig \ref{fig:flowchart}. 
The control action is is chosen by the policy defined via the probability measure 
\begin{align}
    \mathbb{P}_{C}(\mathbf{u}_{t}|\mathbf{a}_{1:t},\mathbf{x}^2_{1:t},\mathbf{u}_{1:t-1}),
\end{align} where we added the subscript $C$ to denote the ``decoder/controller" policy as in Fig \ref{fig:flowchart}. 

We bound the expected length of the prefix-free codeword $\mathbf{a}_{t}$ at every $t$ by bounding the length of  lossless prefix-free source code that encodes $\mathbf{a}_{t}$ itself.  Consider a  modified version of the system model (cf. Fig. \ref{fig:flowchart}) shown in Fig. \ref{fig:sie}. Another, ``virtual"  encoder/decoder pair has been added between the encoder/sensor and decoder/controller. We assume that at every time $t$, the virtual encoder encodes $\mathbf{a}_{t}$ into a prefix-free codeword $\mathbf{c}_{t}$. We refer to $\mathbf{a}_{t}$ as a ``source codeword" and $\mathbf{c}_{t}$ as a ``virtual codeword". At every timestep $t$, the virtual encoder encodes $\mathbf{a}_{t}$ into the virtual codeword $\mathbf{c}_{t}$ via computing
\begin{align}
    \mathbf{c}_{t} = C_{t}(\mathbf{a}_{t}),
\end{align} for some deterministic measurable function $C_{t}$. 
Likewise, the virtual encoder computes the reconstruction $\mathbf{\hat{a}}_{t}$ by computing 
\begin{align}
    \mathbf{\hat{a}}_{t} = D_{t}( \mathbf{c}_{t}),
\end{align}  where again, for all $t$, ${D}_{t}$ is a deterministic measurable function. The virtual encoder and decoder are both memoryless and do not access any SI. We insist that the virtual encoder and decoder are lossless, namely that $\mathbf{a}_{t}\overset{a.s.}{=}
\mathbf{\hat{a}}_{t}$. We think of Fig. \ref{fig:sie} as ``inserting" the virtual encoder and decoder blocks into Fig. \ref{fig:flowchart}. Note that virtual encoder and decoder policies do not effect the measure induced on the random variables $\mathbf{x}_{1:t},\mathbf{a}_{1:t},\mathbf{u}_{1:t}$ due to the assumption that $C_{t}$ and $D_{t}$ are deterministic and that $\mathbf{\hat{a}}_{t}\overset{a.s.}{=}\mathbf{a}_{t}$. 

The idea is that the insertion of an optimal ``virtual" lossless encoder/decoder between the sensor and controller produces a codeword $c_{t}$ that has a length less than or equal to that of $\mathbf{a}_{t}$. More formally, at every time $t$ we lower bound the length of the codeword $\mathbf{a}_{t}$ by lower bounding the length of the codeword $\mathbf{c}_{t}$ under the optimal zero-error prefix free virtual encoder and virtual decoder policies. If $\mathbf{r}$ is a prefix free binary codeword, let $\ell(\mathbf{r})$ denote its length. An ``optimal" virtual encoder and decoder policy (there may be more than one), for some fixed sensor and controller policy, is defined as a sequence of deterministic functions $P^* = \{ {C}^*_{t}, {D}^*_{t} \}$ where:
\begin{enumerate}
    \item ${C}^*_{t}$ maps $\mathbf{a}_{t}$ to $\mathbf{c}_{t}$ and ${D}^*_{t}$ maps $\mathbf{c}_{t}$ to $\mathbf{\hat{a}}_{t}$ as in Fig. \ref{fig:sie}.
    \item There is no probability of error, e.g. $D^{*}_{t}(C^{*}_{t}(\mathbf{a}_{t}))\overset{a.s.}{=}\mathbf{a}_{t}$ for all $t$.
    \item  Let $\mathcal{A}_{t} = \{\mathsf{a}\in \{0,1\}^* : \mathbb{P}(\mathbf{a}_{t}=\mathsf{a})>0\}$. At every time $t$,  if $\mathsf{a}_{1},\mathsf{a}_{2}\in \mathcal{A}_{t}$ and $\mathsf{a}_{1}\neq\mathsf{a}_{2}$ then $\mathsf{c}_{1}=C^*_{t}(\mathsf{a}_{1})$ is not a prefix of $\mathsf{c}_{2} = C^*_{t}(\mathsf{a}_{2})$ and vice-versa.\footnote{This ensures that $C^*_{t}$, restricted to $\mathcal{A}_{t}$ is injective; a necessary condition for there to exist a deterministic $D^*_{t}$ such that $D^*_{t}(C^*_{t}(\mathbf{a}_{t}))\overset{a.s.}{=}\mathbf{a}_{t}$ (cf. condition (2)). It also ensures that the set of virtual codewords transmitted with nonzero probability at time $t$ are not prefixes of one another; define $\mathcal{C}^*_{t} =\{\mathbf{c}\in\{0,1\}^*: \mathbb{P}(C_{t}(\mathbf{a}_{t})=\mathbf{c})>0\}$, and let  $\mathsf{c}_{1},\mathsf{c}_{2}\in\mathcal{C}_{t}$. If $\mathsf{c}_{1}\neq\mathsf{c}_{2}$, then $\mathsf{c}_{1}$ is not a prefix of  $\mathsf{c}_{2}$ and vice-versa.  }
    \item At all $t$, $\mathbb{E}[\ell(\mathbf{c}_{t})]$ is minimized among all other policies satisfying (1), (2), and (3) above.
\end{enumerate} The above expectations and probabilities are taken with respect to the measure induced by the sensor and controller policies $\mathbb{P}_{S}$ and $\mathbb{P}_{C}$ and Fig. \ref{fig:flowchart}. In the language of \cite[Chapter 5]{elemIT}, the constraints (1), (2), and (3) require that at any time $t$,  $C^*_{t}$ is prefix-free code mapping the space $\mathcal{A}_{t}\subset \{0,1\}^*$ to the space of binary prefix-free codewords in $\{0,1\}^{*}$. In the next lemma, we show that under the optimal virtual encoder and decoder policies, we have 
\begin{subequations}\label{eq:virtualHelps}
\begin{align}
    \mathbb{E}[\ell(C^*_{t}(\mathbf{a}_{t}))] \le  \mathbb{E}[\ell(\mathbf{a}_{t})],
\end{align} or, in other words
\begin{align}
    \mathbb{E}[\ell(\mathbf{c}_{t})] \le  \mathbb{E}[\ell(\mathbf{a}_{t})],
\end{align}
\end{subequations}
where $\mathbf{c}_{t}=C^*_{t}(\mathbf{a}_{t})$.
\begin{lemma}
For all $t$ there exists a $C_{t}$ and $D_{t}$ satisfying (1)-(4) above and $\mathbb{E}[\ell(\mathbf{c}_{t})]\le \mathbb{E}[\ell(\mathbf{a}_{t})]$.
\end{lemma}
\begin{proof}
At every time $t$, the source codewords $\mathbf{a}_{t}$ are codewords of a prefix-free code. Setting $C_{t}$ and $D_{t}$ equal to identity, e.g.
\begin{align}\label{eq:app_virtual_encoder}
\mathbf{a}_{t}=C_{t}(\mathbf{a}_{t})&\text{ and } \mathbf{c}_{t}=C_{t}(\mathbf{c}_{t}),
\end{align} gives
\begin{align}
\mathbf{c}_{t}=\mathbf{a}_{t}&\text{ and }\mathbf{\hat{a}}_{t}=\mathbf{c}_{t}
\end{align} 
Under this policy, the virtual encoder sends the input $\mathbf{a}_{t}$ directly and $\mathbf{a}_{t}=\mathbf{c}_{t} =\mathbf{\hat{a}}_{t}$. Thus there exist policies satisfying the constraints that achieve equality in (\ref{eq:virtualHelps}) and the result follows. 
\end{proof} 
Since for every $t$, $C^*_{t}$ is a prefix-free code from $\mathcal{A}_{t}\rightarrow \{0,1\}^{*}$, it follows from \cite[Theorem 5.3.1]{elemIT} that 
\begin{align}
H(\mathbf{a}_{t})\le\mathbb{E}[\ell(C^*_{t}(\mathbf{a}_{t}))],
\end{align} which gives the result
\begin{align}\label{eq:app_conv_conc}
H(\mathbf{a}_{t})\le\mathbb{E}[\ell(\mathbf{a}_{t})].
\end{align} We emphasize that (\ref{eq:app_conv_conc}) holds for every time $t$. 
\section{Proof of equivalence between (\ref{eq:finiteDimensional}) and (\ref{eq:convexProgram})}\label{app:convexification}
We begin by writing (\ref{eq:finiteDimensional}) in terms of $\mathbf{\hat{P}}$ only. It can immediately be seen that the design variables $\mathbf{C}$ and $\mathbf{V}$ are essentially slack. The constraint $\mathbf{\hat{P}}^{-1} = \mathbf{\tilde{P}}^{-1}+\mathbf{C}_{1}^{\mathrm{T}}\mathbf{V}^{-1}\mathbf{C}_{1}$ can be replaced with the constraints $\mathbf{\tilde{P}}-\mathbf{\hat{P}}\succeq \mathbf{0}$ and $\mathbf{\hat{P}}\succ \mathbf{0}$. The new inequality constraint may be readily combined with the equality constraint for $\mathbf{\tilde{P}}$ (cf. (\ref{eq:hat2tilde})) to derive a linear matrix inequality (LMI) in $\mathbf{\hat{P}}$. Applying the matrix inversion lemma to (\ref{eq:hat2tilde}) gives
\begin{multline}\label{eq:almostLMI}
    \mathbf{\tilde{P}}-\mathbf{\hat{P}} =  \mathbf{W}_{11} -\mathbf{\hat{P}}+ \\\mathbf{A}_{11}\left(    \mathbf{\hat{P}}-\mathbf{\hat{P}}\mathbf{A}_{21}^{\mathrm{T}}(\mathbf{A}_{21}\mathbf{\hat{P}}\mathbf{A}_{21}^{\mathrm{T}}+\mathbf{W}_{22})^{-1}\mathbf{A}_{21}\mathbf{\hat{P}}   \right)\mathbf{A}_{11}^{\mathrm{T}}.
\end{multline}  The right hand side of (\ref{eq:almostLMI}) is a Shur complement, and the LMI constraint follows directly. Thus $\mathbf{\tilde{P}}-\mathbf{\hat{P}}\succeq 0$ is equivalent to the LMI
\begin{align}\label{eq:LMI}
    \mathbf{W}+\mathbf{\bar{A}}\mathbf{\hat{P}}\mathbf{\bar{A}}^{\mathrm{T}} -\begin{bmatrix}\mathbf{\hat{P}} & 0\\ 0 & 0\end{bmatrix} \succeq 0.
\end{align} 

The corresponding $\mathbf{C}_{1}$ and $\mathbf{V}$ are not unique, and can be found by factorizing $\mathbf{\hat{P}}^{-1}-\mathbf{\tilde{P}}^{-1}$.

It remains to simplify the rate cost. 
Using (\ref{eq:hat2tilde}) and invoking the matrix determinant lemma twice, we 
\begin{multline}\label{eq:almostFinalSSCost}
    \log\det\mathbf{\tilde{P}} -\log\det\mathbf{\hat{P}} =  \log\det(\mathbf{\hat{P}}^{-1}+\mathbf{\bar{A}}^{\mathrm{T}}\mathbf{W}^{-1}\mathbf{\bar{A}})+\\\log\det\mathbf{W}-\log\det( \mathbf{W}_{22}+\mathbf{A}_{21}\mathbf{\hat{P}}\mathbf{A}_{21}^{\mathrm{T}}).
\end{multline} 
Introduce the slack variable $\boldsymbol{\Pi}$. We have 
\begin{multline}\label{eq:determinentEquivalenceSS}
    \log\det(\mathbf{\hat{P}}^{-1}+\mathbf{\bar{A}}^{\mathrm{T}}\mathbf{W}^{-1}\mathbf{\bar{A}}) =\\ \min_{0\preceq \boldsymbol{\Pi}\preceq  (\mathbf{\hat{P}}^{-1}+\mathbf{\bar{A}}^{\mathrm{T}}\mathbf{W}^{-1}\mathbf{\bar{A}})^{-1}}-\logdet\boldsymbol{\Pi}.
\end{multline} 
Applying the matrix inversion lemma and the Shur complement formula to the constraint $\boldsymbol{\Pi}\preceq  (\mathbf{\hat{P}}^{-1}+\mathbf{\bar{A}}^{\mathrm{T}}\mathbf{W}^{-1}\mathbf{\bar{A}})^{-1}$ gives the equivalent LMI 
\begin{align}
 \begin{bmatrix}
\mathbf{\hat{P}}-\boldsymbol{\Pi} & \mathbf{\hat{P}}\mathbf{\bar{A}}^{\mathrm{T}} \\ \mathbf{\bar{A}}\mathbf{\hat{P}} & \mathbf{W}+\mathbf{\bar{A}}\mathbf{\hat{P}}\mathbf{\bar{A}}^{\mathrm{T}}
\end{bmatrix}\succeq 0. 
\end{align} 

The preceding discussion demonstrates that  
\begin{mini}|s|
{\mathbf{\hat{P}},\boldsymbol{\Pi}}{\frac{\log\det\mathbf{W}-\log\det\boldsymbol{\Pi}-\log\det{
( \mathbf{W}_{22}+\mathbf{A}_{21}\mathbf{\hat{P}}\mathbf{A}_{21}^{\mathrm{T}})}}{2} }
{\label{eq:app_convexProgram}}{}
\addConstraint{\mathbf{\hat{P}}\succ {0}, \boldsymbol{\Pi}\succeq 0, \text{Tr}(\mathbf{\Theta}\mathbf{\hat{P}})+\text{Tr}(\mathbf{S}\mathbf{W}) \le \gamma}
\addConstraint{  \mathbf{W}+\mathbf{\bar{A}}\mathbf{\hat{P}}\mathbf{\bar{A}}^{\mathrm{T}} -\begin{bmatrix}\mathbf{\hat{P}} & 0\\ 0 & 0\end{bmatrix} \succeq 0 }{}
\addConstraint{ \begin{bmatrix}
\mathbf{\hat{P}}-\boldsymbol{\Pi} & \mathbf{\hat{P}}\mathbf{\bar{A}}^{\mathrm{T}} \\ \mathbf{\bar{A}}\mathbf{\hat{P}} & \mathbf{W}+\mathbf{\bar{A}}\mathbf{\hat{P}}\mathbf{\bar{A}}^{\mathrm{T}}
\end{bmatrix}\succeq 0}.
\end{mini}
achieves the same minimum as (\ref{eq:finiteDimensional}). This program is the minimization of a convex objective with convex constraints. 

\section{Proof of Lemma \ref{lem:ditherLemma}}\label{app:ditherLemma} The proof follows closely from \cite{tanakaISIT}. Assume the definitions of Sec. \ref{sec:achieve}. It turns out that $H(\mathbf{\tilde{z}}_{t}|\mathbf{d}_{t} )$
admits a bound in terms of the squared error rate distortion function of $\mathbf{z}$ \cite[Lemma 1 c-d]{tanakaISIT}. Let $D={n\Delta^2/12}=n$. Define the rate distortion function
\begin{align}
    \mathcal{R}_{\mathbf{x}}(D)= \underset{\mathbb{P}(\mathbf{u}|\mathbf{x}):\mathbb{E}[\lVert\mathbf{x}-\mathbf{u}\rVert_{2}^2]\le D}{\inf} I(\mathbf{x};\mathbf{u}).
\end{align}
 We have 
\begin{align}\label{eq:boundCW}
  H(\mathbf{\tilde{z}}_{t}|\mathbf{d}_{t}) \le \frac{n}{2}\log_{2}\frac{4\pi e}{12}+\mathcal{R}_{\mathbf{z}_{t}}(D).
\end{align} It is known (cf. \cite[Problem 10.8]{elemIT}) that if a Gaussian random vector $\mathbf{x}$ has $\text{cov}(\mathbf{x})=\text{cov}(\mathbf{z})$ then  $\mathcal{R}_{\mathbf{z}}(D)\le \mathcal{R}_{\mathbf{{x}}}(D)$. 

We claim that
 \begin{align}\label{eq:rdbound}
     \mathcal{R}_{\mathbf{z}_{t}}(D) \le I(\mathbf{C}_{1}\mathbf{\tilde{r}}^{\mathrm{G}}_{t};\mathbf{C}_{1}\mathbf{\tilde{r}}^{\mathrm{G}}_{t}+\mathbf{v}_{t})
 \end{align} Note that $\boldsymbol{\Theta}_{t}=\mathbf{C}_{1}\mathbf{\tilde{r}}^{\mathrm{G}}_{t}$ is Gaussian and that, by Lemma \ref{lemma:secondMoments} we have $\mathbb{E}[\boldsymbol{\Theta}_{t}\boldsymbol{\Theta}_{t}^{\mathrm{T}}] = \mathbb{E}[\mathbf{z}_{t}\mathbf{z}_{t}^{\mathrm{T}}]=\mathbf{C}_{1}\mathbf{\tilde{P}}_{t}\mathbf{C}_{1}^{\mathrm{T}}$ and that  $\mathbb{E}[\boldsymbol{\Theta}_{t}] =\mathbb{E}[\mathbf{z}_{t}] =  \mathbf{0}$. Thus  $\mathcal{R}_{\mathbf{z}_{t}}(D) \le \mathcal{R}_{\boldsymbol{\Theta}_{t}}(D)$.  By the assumption that $\mathbf{V}=\mathbf{I}$, we have $\mathbb{E}[\mathbf{v}^{\mathrm{T}}_{t}\mathbf{v}_{t}]=n$. Since $D=n$, $\mathcal{R}_{\boldsymbol{\Theta}_{t}}(D)\le I(\mathbf{C}_{1}\mathbf{\tilde{r}}^{\mathrm{G}}_{t};\mathbf{C}_{1}\mathbf{\tilde{r}}^{\mathrm{G}}_{t}+\mathbf{v}_{t})$, and (\ref{eq:rdbound}) follows. Substituting this into this into (\ref{eq:boundCW}) establishes the Lemma.

%%%%%%
%% Appendix:
%% If needed a single appendix is created by
%%
%\appendix
%%
%% If several appendices are needed, then the command
%%
% \appendices
%%
%% in combination with further \section-commands can be used.
%%%%%%

%%%%%%
%% To balance the columns at the last page of the paper use this
%% command:
%%
%\enlargethispage{-1.2cm} 
%%
%% If the balancing should occur in the middle of the references, use
%% the following trigger:
%%
%\IEEEtriggeratref{3}
%%
%% which triggers a \newpage (i.e., new column) just before the given
%% reference number. Note that you need to adapt this if you modify
%% the paper.  The "triggered" command can be changed if desired:
%%
%\IEEEtriggercmd{\enlargethispage{-20cm}}
%%
%%%%%%
%%%%%%
%% References:
%% We recommend the usage of BibTeX:
%%

\end{document}